\documentclass[submission,copyright]{eptcs}

\providecommand{\event}{SOS 2010} % Name of the event you are submitting to

\providecommand{\firstpage}{1}

\usepackage{amsmath,amssymb,stmaryrd,amsthm}
\usepackage{graphicx}
\usepackage{formal, chi, samp, mathpartir}
\usepackage{url}
\usepackage{flafter}
\usepackage{subfig}
\usepackage{tikz}

\usetikzlibrary{shapes,arrows,automata}
\usetikzlibrary{shapes.symbols}
\usetikzlibrary{shapes.misc}
\usetikzlibrary{calc}

\newcommand{\mCIF} {{\text{$\mathrm{hCIF}^{\subset}$}\,}}
\newcommand{\oCIF} {{\text{CIF}\,}}
\newcommand{\paras}{\parallel_{\synca}}

\newcommand{\nvals}[1]{\ensuremath{#1^+}}

\renewcommand{\trans}[3]{\langle #1 \rangle \xrightarrow{#2} \langle #3 \rangle }

\newcommand{\mtrans}[3]{\langle #1 \rangle \stackrel{#2}{\longmapsto} \langle #3 \rangle}

\renewcommand{\ctrans}[3]{\langle #1 \rangle \stackrel{#2}{\dashrightarrow} \langle #3 \rangle}

\newcommand{\eedge}{\ensuremath{\sigma,\sigma'\models_\alpha
    (v,g,a,r,v')}}

\newcommand{\vflow}{\ensuremath{t,\rho\models_\alpha\langle\initf(v),\tcpf(v)\rangle}}
\newcommand{\Sosrule}%
[3]%
{%
  \begin{center}\(
    \sosrule{
      \begin{array}[t]{@{}c}
        #1
      \end{array}
    }
    {
      \begin{array}[t]{@{}c}
        #2
      \end{array}
    }
    {#3}\)
  \end{center}
}

\theoremstyle{plain}
\newtheorem{theorem}{Theorem}[section]

\newtheorem{lemma}[theorem]{Lemma}

\theoremstyle{definition}
\newtheorem{definition}[theorem]{Definition}
\newtheorem{proposition}[theorem]{Proposition}

\newcounter{remarkcounter}
\def\remark{\par\medskip\noindent\refstepcounter{remarkcounter}\hbox{\bf Remark \arabic{remarkcounter}.}
\ %\ignorespaces
}
\def\endremark{\hbox{\bf End of Remark.}
\par\medskip}

\newcounter{excounter}
\def\example{\par\medskip\noindent\refstepcounter{excounter}\hbox{\bf Example \arabic{excounter}.}
\ %\ignorespaces
}
\def\endexample{\hbox{\bf End of Example.}
\par\medskip}

\renewcommand{\theirule}{\arabic{irule}}

\DeclareMathOperator{\initf}{init}
\DeclareMathOperator{\finitf}{\hat{init}}

\newcommand{\actv}{\ensuremath{\mathrm{init}}}

\newcommand{\Vertices}[0]{\ensuremath{\mathcal{L}}}
\renewcommand{\Pred}{\ensuremath{\mathcal{P}}}
\newcommand{\ASpec}{\ensuremath{\mathcal{C}}}
\newcommand{\automata}{\ensuremath{\mathcal{M}}}

\renewcommand{\enc}{\ensuremath{\mathcal{E}}}

\newcommand{\Vars}[0]{\cV}
\newcommand{\values}{\ensuremath{\Lambda}}
\newcommand{\timep}{\ensuremath{\mathbb{T}}}

\newcommand{\Dom}[1]{\ensuremath{\mathrm{dom}(#1)}}

\DeclareMathOperator{\tcpf}{tcp}
\renewcommand{\tcp}{\ensuremath{\mathrm{tcp}}}
\DeclareMathOperator{\ftcpf}{\hat{tcp}}

\DeclareMathOperator{\termf}{term}
\DeclareMathOperator{\ftermf}{\hat{term}}

\newcommand{\tuple}[3]{\ensuremath{(V^{#2}_{#1},\initf^{#2}_{#1}, \tcpf^{#2}_{#1},E^{#2}_{#1},\termf^{#2}_{#1},#3)}}

\newcommand{\htuple}[2]{\ensuremath{(\hat{V}^{#2}_{#1},\hat{\initf}^{#2}_{#1}, \hat{\tcpf}^{#2}_{#1},\hat{E}^{#2}_{#1},\hat{\termf}^{#2}_{#1},\emptyset)}}

\newcommand{\ptuple}[2]{\ensuremath{(V'_{#1},\initf'_{#1}, \tcpf'_{#1},E'_{#1},\termf'_{#1},#2)}}

\newcommand{\pptuple}[2]{\ensuremath{(V''_{#1},\initf''_{#1}, \tcpf''_{#1},E''_{#1},\termf''_{#1},#2)}}

\newcommand{\ipred}[1]{\ensuremath{\mathrm{id}_{#1}}}

\newcommand{\iref}[1]{\ref{#1}^{-1}}

\begin{document}

\def\titlerunning{Hierarchical states in the Compositional Interchange Format}
\def\authorrunning{H. Beohar, D.E. Nadales Agut, D.A.v.Beek, \& P.J.L.Cuijpers}

\title{Hierarchical states in the Compositional Interchange Format}

\author{H. Beohar \quad D.E. Nadales Agut \quad D.A. van Beek \quad P.J.L. Cuijpers
%\institute{Formal methods group\\ Department of Computer science and Mathematics}
\institute{Eindhoven University of Technology (TU/e)\\
P.O. Box 513, NL-5600 MB Eindhoven, The Netherlands
}
\email{\{H.Beohar, D.E.Nadales.Agut, D.A.van Beek, P.J.L. Cuijpers\}@tue.nl}
%\and D.E.Nadales
%\institute{System engineering group\\ Department of Mechanical engineering}
%\institute{Eindhoven University of Technology (TU/e)\\
%P.O. Box 513, NL-5600 MB Eindhoven, The Netherlands
%}
%\email{D.E.Nadales.Agut@tue.nl}
%\and D.A.v. Beek
%\institute{System engineering group\\ Department of Mechanical engineering}
%\institute{Eindhoven University of Technology (TU/e)\\
%P.O. Box 513, NL-5600 MB Eindhoven, The Netherlands
%}
%\email{D.A.v.Beek@tue.nl}
%\and P.J.L. Cuijpers
%\institute{Formal methods group\\ Department of Computer science and Mathematics}
%\institute{Eindhoven University of Technology (TU/e)\\
%P.O. Box 513, NL-5600 MB Eindhoven, The Netherlands
%}
%\email{P.J.L.Cuijpers@tue.nl}
}

\maketitle
\begin{abstract}
\oCIF is a language designed for two purposes, namely as a specification language for hybrid systems and as an interchange format for allowing model transformations between other languages for hybrid systems. To facilitate the top-down development of a hybrid system  and also to be able to express models more succinctly in the \oCIF formalism, we need a mechanism for stepwise refinement. In this paper, we add the notion of hierarchy to a subset of the \oCIF language, which we call \mCIF. The semantic domain of the \oCIF formalism is a hybrid transition system, constructed using structural operational semantics. The goal of this paper is to present a semantics for hierarchy in such a way that only the SOS rules for atomic entities in \mCIF are redesigned in comparison to \oCIF. Furthermore, to be able to reuse existing tools like simulators of the \oCIF language, a procedure to eliminate hierarchy from an automaton is given.
\end{abstract}

\section{Introduction}

Modeling languages for hybrid systems, and hybrid automata in
particular, are designed to combine computational and physical aspects
of a system in one formal model. The compositional interchange format
(CIF), presented in
\cite{Beek5:NewConceptsCIFInpADHS09,CIFSEWIKI:SE10}, is a hybrid
modeling language based on hybrid automata \cite{HenzingerTheory}, but with the semantics defined via structural operational semantics (SOS) \cite{Plotkin-sos-jlap} rules. One of the primary aims of \oCIF is to establish inter-operability among a wide range of tools by means of model
transformations to and from the CIF. In addition, it is possible to specify hybrid systems using CIF, and perform simulation. The reason for using an SOS semantics in an automaton-based framework is that the model transformations
to and from CIF are not only to be executed on `complete' models, but also on components of bigger models. Thus, it is crucial that bisimulation (equivalence) is a congruence for all the constructs of the CIF. This is guaranteed using the process-tyft format of \cite{MousaviRenGro}.

The CIF language contains the following features.
\begin{itemize}
\item Predicates in the locations of the automata that constrain the initial values and/or initial locations ($\mathrm{init}$ predicate), time behavior (invariants and time can progress predicates), and action behavior (invariants).
\item Communication among automata using channels and shared variables.
\item Scoping operators for declaring variables, actions, and channels.
\item An initialization operator for restricting the initial conditions of variables. This allows initialization on a more global level as compared to the $\mathrm{init}$ predicates of automata.
\item A synchronization operator for executing actions synchronously in parallel automata.
\item An urgency operator for declaring actions or channels as urgent.
\end{itemize}
To facilitate the top-down development of a hybrid system in the \oCIF
formalism we need a mechanism for stepwise refinement. In this work we
develop such hierarchical extensions for a subset of CIF called \mCIF,
in which we leave out the constructs for scoping, initialization,
synchronization and urgency. Nevertheless, the semantics presented
here is general enough to allow us to incorporate these concepts in a
straightforward manner. In a later phase, we plan to extend \mCIF to contain these constructs again. This requires us to already take into account some semantic features that are particular for the CIF (like the so-called variable trajectories, guard trajectories and termination trajectories that we discuss further on in this paper) while other features only appear in a reduced form (like the so-called environment transitions that are only used for establishing termination, and not for establishing consistency of a state).

% Stepwise refinement
Stepwise refinement is a framework for designing a system correct by
construction. The following steps are involved in the stepwise
refinement framework as laid out in \cite{multiform}. One starts with
design of a system at a higher level of abstraction, and usually the
model designed at this level is called an abstract model. Then a
concrete model is designed by adding more behavior into the abstract
model such that the concrete model is a refinement of the abstract
model. This process of refining is performed until a desired
implementation is reached. Thus, any formalism that supports stepwise
refinement, must incorporate the following two main things. First, it
should provide a way to add new details in an abstract model.
Secondly, it should provide \textit{at-least} sufficient conditions
under which a concrete model is a refinement of a given abstract model
by construction. In this paper, we consider only the first aspect of
stepwise refinement. %and keep the second one as open for future research.

Consequently, we introduce a notion of hierarchy in \mCIF that allows a straightforward way to add new behavior to a given model. In the past, the following techniques were proposed in order to accommodate stepwise refinement in other formalisms.
\begin{itemize}
\item Action refinement \cite{Glabeek90}. In this approach an action in the alphabet of a process or an automaton is substituted by another process/automaton. However, the setting of action refinement is incompatible with the interleaving models of concurrency as pointed out in \cite{UselSmolka93}. Since the \oCIF and \mCIF formalisms are based on interleaving models of concurrency, we disregard this technique of refinement.
\item %State-tree structures \cite[Chap 3.]{Ma-wonham2005} and
    Statecharts \cite{Harel87-visformalism}. Statecharts were the
    first formalism that extended finite state machines with the
    concept of hierarchy. Conventionally, the semantics of statecharts
    requires a tree-structure on the set of locations of a statechart.
    Consequently, additional concepts from tree-structures, like least
    common ancestors, children of a location, etc., make the semantics
    complicated. We show in the current work that these additional
    concepts are unnecessary when reverting to a structural
    operational semantics. We only need to introduce the notion of a
    substructure. Other concepts of state-structures, like AND-states
    and OR-states, can be expressed through the \emph{parallel
      composition} and the \emph{multiple initial locations} of
    substructures, respectively. The concepts of history retention
    and inter-level transitions (not considered in this paper) are not supported directly in \mCIF,
    but they can be emulated.
%The statetree structure reference is commented as it is a control theory paper and does not fit with the paper on hierarchy (Remark by Pieter.)
\item Hierarchical timed automata
      \cite[Chap 4.]{alexandre-david-hts} are the extensions of statecharts with a finite set of clock variables modeling real time. Again the semantics of this formalism is based on the concepts of tree structures and for this reason we also disregard this approach. However, there is a common intuition about the passage of time in \cite{alexandre-david-hts} with the current work. The time can pass in a hierarchical structure only if the time can pass in all the levels of hierarchy, i.e. time transitions must synchronise in all the levels of hierarchy of a hierarchical automaton.
\item State refinement operator \cite{UselSmolka93}. State refinement is a binary operator on process algebraic processes written as $p[q]$ where $p,q$ are arbitrary process terms. Informally, it means that $p$ is a state with the substructure $q$. In other words, a location of an automaton is allowed to contain another automaton representing its substructure. Furthermore, it was also stated \cite{UselSmolka93} that the above way of introducing hierarchy is compatible with the interleaving models of concurrency. Thus, the present work is motivated by the work carried out in \cite{UselSmolka93}, even though the basic entity in our formalism is an automaton rather than an action.
\end{itemize}

%Next we describe a way to introduce hierarchy in an automaton. We assume a hierarchy function $h$ that returns a composite automaton for a given location. A composite automaton is either an automaton (possibly containing hierarchy itself) or a composition (like parallel composition) of composite automata. We find this hierarchy function to be more suitable for the automata theoretic framework of \oCIF than the state refinement operator \cite{UselSmolka93}, which is more suitable in a process algebraic setting. This is because, unlike in process algebra, the development of the CIF is aimed at modeling convenience rather than at finding the smallest representation of a given construct.

The semantics of the \mCIF formalism is a hybrid transition system
(HTS) \cite{CuijpersReniersHeemels02} constructed using  structural
operational semantics (SOS) \cite{Plotkin-sos-jlap}.
The goal of this paper is to show how the semantics of hierarchical
automata can be defined using SOS rules, in a compositional way. We do
so without introducing the complexity of state tree structures in the
formalism, and in such a way that the semantics of other \oCIF\ operators
remain unaffected. As an additional result, we define an algorithm for
eliminating hierarchy, which enable us to reuse existing tools for
implementing an \mCIF\ simulator.

%We use the word `state' to refer to a state in a HTS and the word `location' to refer to a location/vertex/state of an automaton.
%This work is carried out under the Task T4.4 in work page 4 of MULTIFORM project \cite{multiform}. A notion of refinement (i.e., the second aspect of stepwise refinement) is under development and is required for the completion of the Task T4.4. We anticipate that the current work can be easily extended to the \oCIF. In future, the plan is also to evaluate the complete hierarchical CIF by designing various case-studies within the context of MULTIFORM project.
%%\todo{TODO:1. Link with multiform and old CIF. 2. Future use of \mCIF.}
The remainder of this paper is organized as follows. First, the subset
of \oCIF, which is extended with hierarchy, is presented in
Section~\ref{sec:mcif}. Once the basic concepts are introduced,
Section~\ref{sec:adding-hierarchy-cif} introduces the syntactical
extensions that are needed to add support for hierarchy in \oCIF, and
we discuss the design-decisions that led to such extensions. The
formal semantics of \mCIF\ is presented in
Section~\ref{sec:semantics}, where we illustrate how the concept of
hierarchy can be defined in a compositional and recursive way. In
Section~\ref{sec:elim-hierarchy}, we give a procedure to eliminate the
hierarchy from a \mCIF model. Finally, in Section~\ref{sec:conclusion}
we make some conclusive remarks and discuss future work.

\section{Introduction to \oCIF}
\label{sec:mcif}

This section presents the syntax and semantics of a subset of
\oCIF~\cite{Beek5:NewConceptsCIFInpADHS09}, because presenting the
full syntax and semantics would distract too much from the message of
the paper. Still, we have included all the aspects of the \oCIF that
were involved in the design-decisions we made when defining hierarchy.

As was mentioned in the introduction, the \oCIF language is based on
hybrid automata, which model a combination of computational and
physical behavior of a system by mixing automata theory with the
theory of algebraic differential equations. A key feature of the \oCIF
is that it provides a structural operational semantics for such atomic
(hybrid) automata, which makes the definition of more complex
compositions of automata easier. The only compositions defined in this
paper will be parallel composition and hierarchy, but in
\cite{Beek5:NewConceptsCIFInpADHS09} many more are described.

Informally, a basic \oCIF automaton is shown in Figure
\ref{fig:thermostat} that models the dynamics of a thermostat in a
room. This thermostat can be in one of two computational states. It is
either off, or on. This is reflected by the two circles (called
\textit{locations}) labeled \textit{Off} and \textit{On}. Next to the
label, the locations also contain equations (called
\textit{time-can-progress predicates} or \textit{tcp-predicates}) that
model the physical behavior of the system while it is in this
computational state. In our example, the temperature $T$ will behave
according to the differential equation $\dot{T} = -T + 15$ when the
thermostat is off and according to $\dot{T} = -T + 25$ when the
thermostat is on. The dotted version of variables, like $\dot{T}$ in
the example, are used for modeling derivatives.
The execution of a calculation generally results in a change of
location, which is modeled by an arrow (called \textit{edge}) from one location to another.
Edges are labeled by \textit{actions} (in our example
\textit{switch-on} and \textit{switch-off}) that may be used to
synchronize the behavior of automata in a composition (not shown
here). Furthermore,
they contain a predicate (called \textit{guard}) that determines under
which condition an action can be executed, and a predicate (called
\textit{reset}) that determines whether there is a change in any of
the model variables. In case of the thermostat, the \textit{switch-on}
action can only be executed if the temperature $T$ is lower than $20$,
and the action results in a change of the variable $n$ to $n+1$, which
counts the number of times the thermostat is switched on.
Notation $n^+$ is used to refer to the value of $n$ after the
execution of the action.
 The guard $n
\leq 1000$ disables the \textit{switch-off} action,  modeling that the
thermostat breaks down after a thousand switches and leaves the room
hot.
Every location contains an \emph{initialization predicate}, which
determines whether execution can start in that location. In the
example, initially the thermostat is switched off, the temperature in
the room is $T = 25$ and the counter is set to $n = 0$. This is
modeled by an incoming edge in the \textit{Off} location without any
origin, labeled by the predicate $T=25 \wedge n=0$. The absence of
incoming arrows on location $\mathit{Off}$ denotes that the
initialization predicate is $\False$ (which means that execution
cannot start on that location).
%%
%%  Avoid references to semantic concepts whenever possible
%%
%The term \emph{active locations} is
%used to refer to the set of locations for which the initialization
%predicate holds. After an action is performed the a unique active
%location is picked.

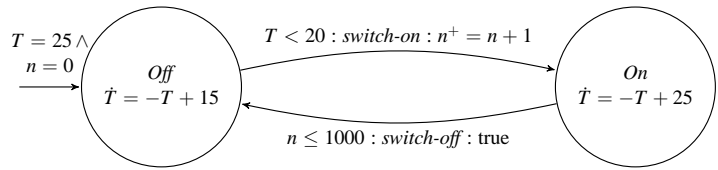
\begin{figure}
    \centering
    \scalebox{0.7}{
    \tikzset{
  super/.style={
    rectangle,
    rounded corners,
    draw=black,
  }
}

\tikzset{
  active/.style={
    draw=red,
  }
}

\tikzset{
  state/.style={
    circle,
    draw=black,
  }
}
\begin{tikzpicture}[->,>=stealth',shorten >=1pt,auto, semithick]

\node (i) {};
\node (off) [state] at ($(i.east)+(2.7,0)$) {$\begin{array}{c}
\mathit{Off} \\ \dot{T}=-T+15
\end{array}$};
\node (on) [state] at ($(off.center) + (9,0)$) {$\begin{array}{c}
\mathit{On} \\ \dot{T}=-T+25
\end{array}$};

\path (i) edge node {$\begin{array}{c}
T=25\, \wedge\\ n=0
\end{array}$} (off);
\path (off) edge  [bend left=12] node {$T < 20 :\textit{switch-on}:
n^{+}=n+1$} (on);
\path (on) edge [bend left=12] node {$ n \leq 1000 :\textit{switch-off}:
\True$} (off);

\end{tikzpicture}
%%% Local Variables:
%%% mode: latex
%%% TeX-master: "../m-HCIF"
%%% End:
}
    \caption{A model of thermostat.}\label{fig:thermostat}
\vspace{-1ex}
\end{figure}

Formally, the
locations of an atomic \oCIF\  automata are taken from the set
$\Vertices$.
 Actions
belong to the set $\actions$. We distinguish the following types of
variables: regular variables, denoted by the set $\Vars$; the dotted
versions of those variables, which belong to the set $\dot{\Vars}
\triangleq \{ \dot{x}\ | \ x \in \Vars \}$; and \emph{step variables},
which belong to the set $ \{ x^+ \ |\ x \in \Vars \cup \dot{\Vars}
\}$.
% Discrete and continuous variables
Furthermore, the variables can be classified according their
continuous evolution (i.e. how their values change during time
delays). In particular, we distinguish between \emph{discrete
  variables} ($n$ in the previous example), whose values remain
constant during time delays, and the value of their dotted versions
are always $0$; and \emph{continuous variables} ($T$ in the previous
example), whose values evolve as a continuous function of time during
delays, and whose dotted versions represent their derivatives. Variables
are constrained by differential algebraic equations and we implement
them as predicates. The values of the variables belong to the set
$\Lambda$ that contains, among else, the sets $\setofbooleans$,
$\setofreals$, and $\setofcomplex$. Guards, tcp and initialization predicates, and reset
predicates are taken from the sets $\Pred_g$, $\Pred_t$, and
$\Pred_r$, respectively.

% Discussion about predicates
The exact syntax and semantics of predicates are left as a parameter
of our theory, as we are not interested in the computational aspects
of \oCIF in this paper. In the examples presented here, and in the
tool implementations of CIF, $\Pred_g$, $\Pred_t$, and $\Pred_r$ are
terms of the language of predicate logic
\cite{Reynolds:TheoriesOfProgrammingLanguages99}, where for $\Pred_g$
and $\Pred_t$ the variables are taken from the set $\Vars \cup
\dot{\Vars}$, and for $\Pred_r$ the variables are taken from the set
$\Vars \cup \dot{\Vars} \cup \{ x^+ \ |\ x \in \Vars \cup \dot{\Vars}
\}$.

Given these preliminaries, an atomic automaton can be defined as follows.
\begin{definition}[Atomic automaton] An atomic automaton is a tuple
  $(V, \actv, \tcp, E)$ with a set of locations $V \subseteq
  \Vertices$; initial and time-can-progress predicates $\actv$, $\tcp
  \colon V \rightarrow \Pred_t$; and a set of edges $E \subseteq V
  \times \Pred_g \times \actions \times {\Pred}_r \times
  L$.
\end{definition}

% Compositions
We use symbol $\automata$ to refer to the set of all atomic automata.
Atomic automata, as the one shown before, can be used to build more
complex models by using the parallel composition operator. \oCIF
includes more operators, but we do not discuss them in this paper.
Throughout this work, we use the term \emph{composition} to refer
either to an atomic automata, or to a parallel composition of
automata, which is denoted as $p \parallel q$, for compositions $p$
and $q$, where the set $S\subseteq \actions$ is the set of actions
that must be executed synchronously in both automata. The set
$\ASpec$ contains all \mCIF\ compositions, and is formally defined in
Section~\ref{sec:adding-hierarchy-cif}.

% Outro:
It is not possible to present the formal semantics of \oCIF\ in this
paper due to the lack of space. However, if an automaton has no
hierarchy (i.e. no location contains a composition), the rules
presented here match those of \oCIF.

After introducing the base language, we are ready to show how it can
be extended with hierarchy.
%  A composite automaton is either an automaton
% (possibly containing hierarchy itself) or a composition (like parallel
% composition) of composite automata.

%%
%% I don't think it is necessary...
%%
% For instance, the
% synchronising operator specifies which actions of an automaton must be
% made synchronising, urgent operators specifies urgent actions of an
% automaton, parallel composition allows concurrent execution of two or
% more automata. However, in this paper (i.e. within the context of
% hierarchy) we only consider the parallel composition operator of
% \oCIF.

\section{Adding hierarchy to CIF}
\label{sec:adding-hierarchy-cif}

% Intro
In this section, we show the syntactical extensions that are needed
for adding hierarchy to \oCIF, explaining why every new element is
required according to the design-decisions we have taken.

% What is and how do we incorporate hierarchy
An automaton is said to be \emph{hierarchical} if it contains a
composition in at least one of its locations. Extending \oCIF\ with
hierarchy is easy to achieve with the addition of a hierarchy function
$h$ to the elements of an automaton, such that $h(v)$ returns the
composition contained in $v$, which we refer to as \emph{substructure}.
% Why not using state refinement?
We find this hierarchy function to be more suitable for the automata
theoretic framework of \oCIF\ than the state refinement operator
\cite{UselSmolka93}, which was conceived for process algebraic
setting. This is because, unlike in process algebra, the development
of the CIF is aimed at modeling convenience rather than at finding the
smallest representation of a given construct.

% Example of hierarchical automaton
As an example, suppose we want to extend the model of the thermostat
presented in Figure~\ref{fig:thermostat}, so that it is only switched
off after a certain time has elapsed (which allows the room to be
heated up). Having hierarchy, a refined model of the thermostat can be
elaborated as in Figure \ref{fig:hierarchical-thermostat}. 
We define the hierarchy function $h$ such that $\mathit{Off} \notin \Dom{h}$; 
and $h(\mathit{On})$ returns the automaton shown at the bottom in Figure~\ref{fig:thermostat-clock}, 
which initially sets up a clock $c$ (continuous variable), and it waits until the timer expires $\Delta
\leq c$ (where $0 < \Delta$) to generate the event $\mathit{done}$.

% The thing inside the thing
\begin{figure}[htb]
  \centering
  \scalebox{0.7}{
    \tikzset{
  super/.style={
    rectangle,
    rounded corners,
    draw=black,
  }
}

\tikzset{
  active/.style={
    draw=red,
  }
}

\tikzset{
  state/.style={
    circle,
    draw=black,
  }
}
\begin{tikzpicture}[->,>=stealth',shorten >=1pt,auto, semithick]

\node (i) {};
\node (off) [state] at ($(i.east)+(2.7,0)$) {$\begin{array}{c}
\mathit{Off} \\ \dot{T}=-T+15
\end{array}$};
\node (on) [state] at ($(off.center) + (9,0)$) {$\begin{array}{c}
\mathit{On} \\ \dot{T}=-T+25
\end{array}$};

\path (i) edge node {$\begin{array}{c}
T=25\, \wedge\\ n=0
\end{array}$} (off);
\path (off) edge  [bend left=12] node {$T < 20 :\textit{switch-on}:
n^{+}=n+1$} (on);
\path (on) edge [bend left=12] node {$ n \leq 1000 :\textit{switch-off}:
\True$} (off);

\node (begin) at ($(on.center) + (-5,-5)$) {};
\node (cold) [state] at ($(begin.east)+(2,0)$) {$\begin{array}{c}
\mathit{Cold} \\ \dot{c}=1
\end{array}$};
\node (hot) [state] at ($(cold.center) + (6,0)$) {$\begin{array}{c}
\mathit{Hot} \\ \True
\end{array}$};
\node (end)  at ($(hot.east)+(1,0)$) {};

\path (begin) edge node {$\begin{array}{c}
c=0
\end{array}$} (cold);
\path (cold) edge node {$\Delta \leq  c :\textit{done}:
\True$} (hot);
\path (hot) edge node {$
  \begin{array}{c}
    \True
  \end{array}
$} (end);

\path ($(on.west)+(1,-1)$) [-,draw,dashed] to ($(begin.west)+(0,1)$);
\path ($(on.east)+(-1,-1)$) [-,draw,dashed] to ($(end.east)+(0,1)$);

\end{tikzpicture}
%%% Local Variables:
%%% mode: latex
%%% TeX-master: "../text"
%%% End:
}
  \caption{Hierarchical model of the thermostat.}
  \label{fig:hierarchical-thermostat}   \label{fig:thermostat-clock}
\end{figure}
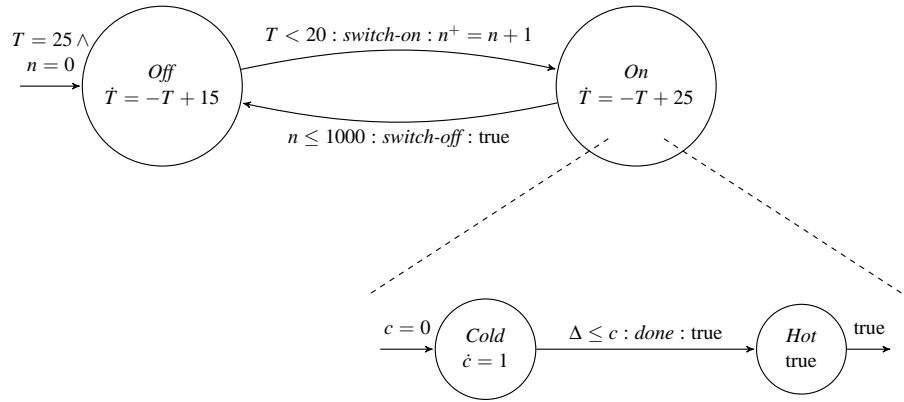

An important thing to note about the substructure in figure \ref{fig:thermostat-clock},
is that one of the states contains an outgoing arrow that leads nowhere.
The predicate on this arrow is called a \emph{termination predicate}
and it is taken, just like tcp and initialization predicates, from the
set $\Pred_t$.

Termination predicates were not part of the atomic automata of CIF before,
even though they are common in the general theory of automata. Our reason
for adding them is that we need a mechanism to decide when a substructure
hands over the control of the execution flow to the superstate to which it
belongs.

In most hierarchical formalisms, such as \cite{Uselton1994} and \cite{Lee:EECS-2009-151},
the actions enabled in the super-automaton are executed regardless of the state of the
substructure. The example in Figure \ref{fig:thermostat-clock}, however, illustrates that
a more general approach, in which the substructure has control over the superstructure, may be useful.
The example actually uses the fact that the substructure has control over the superstructure to restrict
the behavior of the thermostat in such a way that it is forced to stay in the $\mathit{On}$ location for a certain time.
The termination mechanism that we have chosen, originates from our desire to be able to (partially)
translate sequential function charts \cite{Bauer04s.:a} to CIF. In that formalism, termination
is used as the standard mechanism to pass control from one chart to the next.

Another mechanism that we need when dealing with hierarchy, is a way to keep
track of the so-called \emph{active location} that an automaton is in. Admittedly,
the active location of an automaton is a semantic concept, used to keep track of
the current state when describing the dynamics of an automaton, and it does not belong
in a section describing the syntax. However, in the structured operational semantics
of the CIF, the states are formed by pairs of syntactic descriptions and valuations of
variables (see section \ref{sec:semantics} for details), and our solution that deals with
the semantic problem of keeping track of the active location, makes use of an auxiliary
syntactic construct that one should not use while modeling, but that formally is part of
the syntax of CIF.

In the non-hierarchical CIF, the active location of an automaton is fully determined by
its initialization predicate. As a result, a state-change in the semantics is modeled
by changing this initialization predicate. In the hierarchical CIF, however, changing
the initialization predicates of a substructure means changing the hierarchy function $h$,
and this causes so-called history retention; i.e.\ if a substructure is terminated and the
automaton gets back to it later, the substructure will restart where it was ended previously,
because $h$ is still in its altered form. From the example in Figure \ref{fig:hierarchical-thermostat},
one can see that this is not always desirable. History retention in that example, would make
that the thermostat is only forced to linger in the $\mathit{On}$ state the first time
it is entered. In subsequent visits, the substructure would already be in its terminating
state immediately.

The semantics we would like to give to hierarchical CIF, is that a substructure is
restarted every time from one of its (original) initial states. Our solution for the
history retention problem, is to introduce an auxiliary composition operator $p:\alpha$
which should be read as $\alpha$ is currently in the substructure $p$. Here, $p$ is an
arbitrary composite automaton and $\alpha$ is an atomic hierarchical automaton. Next,
the initialization predicate of $\alpha$ can be used to model the active location of the
super-automaton $\alpha$, while the initialization predicate(s) of (the components) $p$
are used to model the active location of the active substructure. We found that
the use of this auxiliary operator greatly simplifies the structured operational
semantics of hierarchical CIF.

Now that we have informally introduced all the syntactic elements needed for
extending \oCIF\ with hierarchy, we define the syntax of hierarchical automata
as follows.

\begin{definition}\label{def:hierarchical-automata}
  An atomic hierarchical hybrid automaton is a tuple $( V,\initf, \tcp, E,
  \termf, h)$ where, $(V, \initf, \tcpf, E)$ is an atomic CIF automaton,
  $\termf: V\rightarrow \Pred$ is a function that associate to each
  location a predicate describing the conditions under which a
  location is final, and $h: V \rightharpoonup \ASpec$ is a (partial)
  hierarchy function that maps each location with a composite
  automaton. The set of composite automata $\ASpec$ in \mCIF is
  recursively defined by the following grammar
  $\ASpec::=  \alpha \mid \ASpec:\alpha \mid \ASpec \paras \ASpec$.
\end{definition}

Henceforth, we use Greek letters $\alpha$ and $\alpha'$ to indicate an atomic hierarchical hybrid automaton
and Roman letters $p$, $q$, $p'$, and $q'$ to indicate any composite automaton in $\ASpec$.

\section{Formal semantics of \mCIF}
\label{sec:semantics}
In this section we illustrate how the concept of hierarchy can be
defined in a compositional way, without introducing additional
concepts of tree-structures present in the statechart
\cite{Harel87-visformalism} formalism.
First, the semantic framework is set up, and then the SOS rules are
presented.

The semantics of \mCIF\ compositions (and \oCIF) is given in terms of
SOS rules, which induce hybrid transition systems (HTS)
\cite{CuijpersReniersHeemels02}. The states of the HTS are of the form
$\langle p, \sigma \rangle$, where $p \in \ASpec$ is a composition and
$\sigma: \Vars \cup \dot{\Vars}\rightarrow \values$ is a function,
called \emph{valuation}, which assigns values to the variables.
Valuations capture the phenomenon of discrete change in the values of
variables caused by the execution of actions in an automaton. We
denote the set of all valuations as $\Sigma$. There are three kind of
transition in the HTS, namely, \emph{action transitions},
\emph{environment transitions}, and \emph{time transitions}. We
describe them in detail next.

Action transition are of the form $\trans{p, \sigma}{a}{p', \sigma'}$,
and they model the execution of action $a$ by process $p$ in an
initial valuation $\sigma$, which changes process $p$ into $p'$ and
results in a valuation $\sigma'$.

Environment transitions are of the form $\ctrans{p,\sigma}{b}{p',\sigma'}$,
and in the full CIF language, they are used to model which possible behavior of
the environment is consistent with that of the composition $p$, but cannot be
executed by the component itself. In the restricted language \mCIF, the function
of the environment transitions is to indicate that a composition $p$ can initialize
to become a composition $p'$ in which an active location is fixed for each (active)
substructure. Furthermore, the boolean $b$ indicates whether the initialized substructure
can terminate, and thus give back the control over actions to their environment.

Time transitions are of the form $\mtrans{p, \sigma}{\rho, \theta,
  \omega}{p', \sigma'}$, and they model the passage of time in
composition $p$, in an initial valuation $\sigma$, which results in a
composition $p'$ and valuation $\sigma'$. The relation between $p$ and
$p'$ is the same as for environment transitions. Function $\rho:
\timep \rightarrow \Sigma$ is called \emph{variable trajectory}, and
it models the evolution of variables during the time delay. For each
time point $s \in \Dom{\rho}$, and for each variable $x \in \Vars \cup
\dot{\Vars}$, the function application $\rho(s)(x)$ yields the value
of variable $x$ at time $s$. Function $\theta:\timep\rightarrow
2^{\actions}$ is called \emph{guard trajectory}, and it models the
evolution of enabled actions during time delays. For each time point
$s\in \Dom{\theta}$, the function application $\theta(s)$ yields the
set of enabled actions of composition $p$ at time $s$. Lastly,
function $\omega$ is called \emph{termination trajectory}, and it
models the evolution of termination during time delays: for each time
point $s\in \Dom{\omega}$, composition $p'$ is terminating at time $s$
if and only if $\omega(s)$. For all time transition
$\Dom{\rho}=[0,t]$, for some time point $t\in \timep$, and
$\Dom{\rho}=\Dom{\theta}=\Dom{\omega}$.

Guard trajectories were shown to be necessary for the definition of
urgency and variable abstraction in CIF and other hybrid formalisms
\cite{CuijpersReniers08,BeekCuijpersMarkovskiNadalesRooda10}.
Termination trajectories allow us to keep track of the possibility of
termination over time, and they are essential for constructing the
guard trajectories in the rules. Even though these concepts are not
necessary for giving semantics to \mCIF, they allow us to solve the
problem of supporting urgency and variable abstraction in a
hierarchical setting, and thus our approach can be extended to the
whole \oCIF without modifying the rules.

Even though predicates are abstract entities, we assume that there is
a satisfaction relation $\sigma \models e$ is defined, which expresses
that predicate $e \in \Pred_t \cup \Pred_g \cup \Pred_r$ is satisfied
(i.e. it is true) in valuation $\sigma$. For predicate logic, this
relation can be defined in a standard way (see
\cite{Reynolds:TheoriesOfProgrammingLanguages99} for example). For a
valuation $\sigma$, we define $\sigma^+\triangleq\{ (v^+, c) \ | \ (v, c) \in
\sigma \}$.

Definition~\ref{def:hts} formalizes the hybrid transition system
induced by the SOS rules presented in the next sections.

\begin{definition}\label{def:hts}
  A hybrid transition system (HTS) is a six-tuple of the form
  $(Q,\actions,\xrightarrow{},\longmapsto,\dashrightarrow)$ where,
  $Q\triangleq\ASpec\times\Sigma$, $\xrightarrow{}\subseteq
  Q\times\actions\times Q$, $\longmapsto\subseteq Q\times\left(
    (\timep\rightarrow\Sigma) \times (\timep\rightarrow 2^{\actions})
    \times (\timep\rightarrow \setofbooleans)\right)\times Q$, and
  $\dashrightarrow\subseteq Q\times \setofbooleans \times Q$.
\end{definition}

\subsection{Hierarchical hybrid automaton}
\label{sec:atomic-automata}

In this section, we give semantics to hierarchical hybrid automata. We
use notation $\alpha$ to refer to an atomic automaton $(V, \actv,
\tcpf, E, \termf, h)$, and $\alpha{[v]}$ to refer to the automaton
$(V, \ipred{v}, \tcpf, E, \termf, h)$, where $\ipred{v}(w)\triangleq v=w$. % We write $\sigma \models e$ to
% denote that a valuation $\sigma$ satisfies a predicate $e \in
% \Pred_i$, with $i \in \{g, t, r\}$.

In absence of hierarchy, an atomic automaton $\alpha$ can perform an
action in a location $v$ and initial valuation $\sigma$ if there is an
edge $(v, g, a, r, v')$ such that the following conditions hold:
\begin{enumerate}
\item Location $v$ is active ($\sigma \models \actv(v)$).
\item Guard $g$ holds ($\sigma \models g$).
\item It is possible to find a new valuation $\sigma'$ such that the
  reset predicate is satisfied in valuation $\sigma \cup {\sigma'}^+$
  ($\sigma \cup {\sigma'}^+ \models r$). We do not write $\sigma'
  \models r$ since, in general, $r$ refers to the next values of
  variables, which are contained in ${\sigma'}^+$.
\end{enumerate}
The above conditions are summarized in the term $\eedge$, which is
syntactically equivalent to condition $(v,g,a,r,v') \in E\wedge
\sigma\models\initf(v) \wedge \sigma\models g\wedge \nvals{\sigma'}
\cup \sigma \models r\enspace$.

Rule~\ref{rule:action:automaton:sa-term} describes what is
semantically involved with the addition of hierarchy. Firstly, it is
necessary to check that the substructure of the initial location, if
any, is terminating (condition $\ctrans{h(v),\sigma}{\true}{p,\sigma}
\vee v\not\in\Dom{h}$). Finally, after the action is performed, the
substructure in the target location, if present, must be initialized
(condition $\trans{\alpha,\sigma}{a}{q:\alpha[v'],\sigma'}$). Note
that the choice of selecting active locations of substructure $h(v')$
is made upon entering location $v'$.
Example~\ref{example:eager-choice} illustrates this phenomenon, which
we call \emph{eager choice}.

\[
\irule
{
\eedge,
\Big(\ctrans{h(v),\sigma}{\true}{p,\sigma} \vee v\not\in\Dom{h}\Big),\\\\
v'\in\Dom{h},\ctrans{h(v'),\sigma'}{b}{q,\sigma'}
}
{
\trans{\alpha,\sigma}{a}{q:\alpha[v'],\sigma'}
}\label{rule:action:automaton:sa-term}
\]

\begin{example}\label{example:eager-choice}
  Consider the composite automaton shown in
  Figure~\ref{fig:parallel-automata} in which the $\tcpf$ function
  gives the predicate $\true$ for all the vertices of the automaton.
  The idea behind eager choice, is that after the execution of the
  action $a$, an initial state of the substructure is picked
  immediately. This can only result in the left state of the
  substructure to be picked, because of the value of $x$ is set to
  $1$ during the execution of $a$. Hence, the action $c$ in the
  substructure will never be executed. The resulting transition
  system generated by the SOS is shown in
  Figure~\ref{fig:eager-parallel-automata}, where the states are
  depicted as circles and their components are not shown.
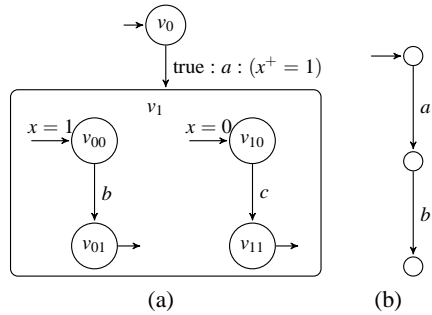
\begin{figure}[!ht]
  \centering
  \subfloat[][]{\label{fig:parallel-automata}
  \scalebox{0.7}{\tikzset{
  super/.style={
    rectangle,
    rounded corners,
    draw=black,
  }
}

\tikzset{
  active/.style={
    draw=red,
  }
}

\tikzset{
  state/.style={
    circle,
    draw=black,
  }
}
\begin{tikzpicture}[->,>=stealth',shorten >=1pt,auto, semithick]

\node (i) {};
\node (v0) [state] at ($(i.east)+(0.8,0)$) {$v_0$};
\node (v1) [super,matrix] at ($(v0.center)+(0,-3)$) {
    \node (i0) {};
    \node (name) at ($(i0.center) + (2.5,0.6)$) {$v_1$};
    \node (i1) at ($(i0.center) + (3,0)$) {};
    \node (v00) [state] at ($(i0.east)+(1.2,0)$) {$v_{00}$};
    \node (v10) [state] at ($(i1.east)+(1.2,0)$) {$v_{10}$};
    \node (v01) [state] at ($(v00.center)+(0,-2)$) {$v_{01}$};
    \node (v11) [state] at ($(v10.center)+(0,-2)$) {$v_{11}$};
    \node (f01)  at ($(v01.east)+(0.6,0)$) {};
    \node (f11)  at ($(v11.east)+(0.6,0)$) {};\\
};

\path (i) edge (v0);
\path (v0) edge node {$\true:a:(x^{+}=1)$} (v1);
\path (i0) edge node {$x=1$} (v00);
\path (i1) edge node {$x=0$} (v10);
\path (v00) edge node {$b$} (v01);
\path (v10) edge node {$c$} (v11);
\path (v01) edge (f01);
\path (v11) edge (f11);

\end{tikzpicture}
%%% Local Variables:
%%% mode: latex
%%% TeX-master: "../scif_ppt"
%%% End:
}}
  \subfloat[][]{\label{fig:eager-parallel-automata}
  \scalebox{0.7}{\tikzset{
  super/.style={
    rectangle,
    rounded corners,
    draw=black,
  }
}

\tikzset{
  active/.style={
    draw=red,
  }
}

\tikzset{
  state/.style={
    circle,
    draw=black,
  }
}
\begin{tikzpicture}[->,>=stealth',shorten >=1pt,auto, semithick]

\node (i) {};
\node (v0) [state] at ($(i.east)+(0.8,0)$) {};
\node (v1) [state] at ($(v0.center)+(0,-2)$) {};
\node (v2) [state] at ($(v1.center)+(0,-2)$) {};
%\node (v3) [state] at ($(v0.center)+(2,-2)$) {$x=1$};
%\node (v4) [state] at ($(v3.center)+(0,-2)$) {};

\path (i) edge (v0);
\path (v0) edge node {$a$} (v1);
\path (v1) edge node {$b$} (v2);
%\path (v0) edge node {$a$} (v3);
%\path (v3) edge node {$c$} (v4);
%\path[-] ($(v3.west) + (0.03,0.45)$) edge node {} ($(v3.east) + (-0.04,-0.45)$);

%\path[-] ($(v3.west) + (0.03,-0.45)$) edge node {} ($(v3.east) + (-0.04,0.45)$);

\end{tikzpicture}
%%% Local Variables:
%%% mode: latex
%%% TeX-master: "../scif_ppt"
%%% End:
}}
  \caption{\subref{fig:parallel-automata} Automaton with two
    possible initial states. \subref{fig:eager-parallel-automata}
    Resulting HTS.}
\end{figure}
\end{example}

Rule \ref{rule:action:automaton:sa-term} requires as a condition that
there is an active substructure in the target location $v'\in\Dom{h}$.
If this is not the case then no active substructure is prefixed to
$\alpha[v]$, as expressed by
Rule~\ref{rule:action:automaton:target-h-nd}.

\[
\irule
{
\eedge\com v'\not\in\Dom{h},\\\\
\Big(\ctrans{h(v),\sigma}{\true}{p,\sigma} \vee v\not\in\Dom{h}\Big)
%check for h(v),h(v') defined
}
{
\trans{\alpha,\sigma}{a}{\alpha[v'],\sigma'}
}\label{rule:action:automaton:target-h-nd}
\]

In a hierarchical setting, actions in an automaton can be generated by
the substructure of an active location.
Rule~\ref{rule:action:automaton:source-sa-nterm} formalizes this. Note
that in the conclusion, $p:\alpha[v]$ reflects the fact that an
initial location is chosen in a hierarchical structure if the
substructure performs an action.

\[
\irule
{
\sigma\models\initf(v), v\in\Dom{h},
\trans{h(v),\sigma}{a}{p,\sigma'}
}
{
\trans{\alpha,\sigma}{a}{p:\alpha[v],\sigma'}
}\label{rule:action:automaton:source-sa-nterm}
\]

In \oCIF, a time delay is possible in an active location $v$ if there
exists a trajectory $\rho$ such that the tcp predicate is satisfied in
$[0, t)$. Henceforth, we will use $\vflow$ as an abbreviation of the
predicate $\rho(0) \models \initf(v)\wedge \Dom{\rho} = [0,t]\wedge 0
\leq t \wedge \forall s\in[0,t).[\rho(s) \models \tcp(v)]\enspace$.

% When time passes in $v$, should it also pass in $h(v)$? Moreover,
% since $h(v)$ can contain another active locations, there is a tcp
% predicates associated with each one of them. Then, which tcp
% predicates are to be considered during the time delay?

For time delays, in \mCIF the substructure must perform a time
transition with the same trajectory, and we consider conjunction of
all the tcp predicates of all the active locations of an automaton. In
this way time passes in an automaton, and also in all of its contained
substructures. This is, an automaton and its substructure
synchronize in the time delays. In the complete extension of CIF with
hierarchy, a similar approach is taken for invariants.
Rule~\ref{rule:time:automaton} models this, where
$\Dom{\omega}=\Dom{\rho}\com \Dom{\theta} = \Dom{\rho}$, $\forall_{s
  \in [0, t]}. \omega(s) = \omega_0(s) \wedge
\rho(s)\models\termf(v)$, and $\forall_{ s\in[0,t]}.\theta(s)=
\theta_0(s) \cup \{a \mid (v,g,a,r,v') \in E \wedge \rho(s)\models g
\wedge \omega_0(s)\}$. The guard trajectory $\theta$ as well
as the termination trajectory $\omega$ are constructed by using the
corresponding trajectories generated by the time transition in the
substructure.

We found that this approach is the simple and intuitive: substructures
are part of the whole structure, and it is strange from the modeling
point of view to have a system in which time can ``freeze'' for
certain parts of the system. Furthermore, since there are several
active locations in different levels of hierarchy, it is not clear
which one to choose to perform the time delays. The example of the
thermostat with hierarchy, depicted in Figure~\ref{fig:hierarchical-thermostat}, shows the convenience of this decision,
since we want the clock to advance while the room heats up.

\[
\irule
{
\vflow\com
\mtrans{h(v),\rho(0)}{\rho,\theta_0,\omega_0}{p,\rho(t)}
}
{
\mtrans{\alpha,\rho(0)}{\rho,\theta,\omega}{p:\alpha[v],\rho(t)}
}\label{rule:time:automaton}
\]

The following example illustrates the need for having termination
trajectories to capture properly the set of enabled actions during
time delays.
\begin{example}
Consider the automaton shown in Figure~\ref{fig:enabled-actions} and assume $1<k_0<k_1$. Then the set of enabled actions at the active location will depend on the function $e^x$ and this set is illustrated in Figure~\ref{fig:example-enabled-actions}. In other words, if $0\leq x< \ln(k_0)$ then the set of enabled actions is $\{a\}$. If $\ln(k_0)\leq x <\ln(k_1)$ then the set of enabled actions is $\{a,b\}$. And if $x\geq \ln(k_1)$ then the set of enabled actions is $\{a\}$.
\end{example}
\begin{figure}[!ht]
\centering
\subfloat[]{
\scalebox{0.9}{
\tikzset{
  super/.style={
    rectangle,
    rounded corners,
    draw=black,
  }
}

\tikzset{
  active/.style={
    draw=red,
  }
}

\tikzset{
  state/.style={
    circle,
    draw=black,
  }
}

\begin{tikzpicture}[->,>=stealth',shorten >=1pt,auto, semithick]
  % Damn pgf!
  % at ($(#1.east)+(#1,0)$)
  \node (v) [super,matrix] {
    \node (i) {};
    \node (v0) [state] at ($(i.east)+(0.8,0)$) {$v_0$};
    \node (v1) at ($(v0.center)+(0,1)$) {$\cdots$};
    \node (f)  at ($(v0.east)+(1.2,0)$) {$k_0< e^x$};\\
  };
  \node (v2) at ($(v.center) + (4,0)$) {};
  
  \path (i) edge (v0);
  \begin{tiny}
    \path (v0) edge node {$true,a$} (v1);
  \end{tiny}
  \path (v0) edge (f);
  \path (v) edge node {$e^x<k_1,b$} (v2);

\end{tikzpicture}

%%% Local Variables:
%%% mode: latex
%%% TeX-master: "../testslides"
%%% End:
}
\label{fig:enabled-actions}}\quad
\subfloat[]{
\includegraphics[width=0.3\textwidth]{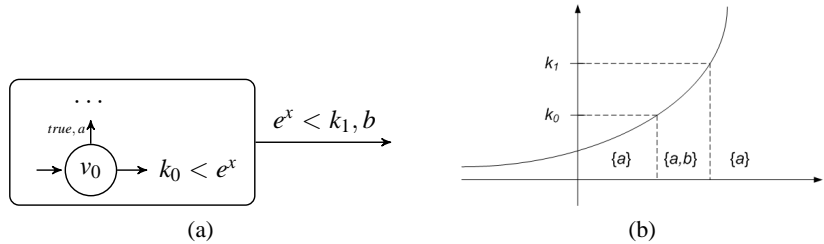}
\label{fig:example-enabled-actions}}
\caption{Illustration of the dependence of enabled actions over time.}
\end{figure}

Rule~\ref{rule:time:automaton:target-h-nd} deals with the case when an
initial location $v$ does not contain a substructure, where
$\Dom{\omega}=\Dom{\rho}\com \Dom{\theta} = \Dom{\rho}$ and
$\forall_{s \in [0, t]}. \omega(s) = \rho(s)\models\termf(v)$, and
$\forall_{ s\in[0,t]}.\theta(s)= \{a \mid (v,g,a,r,v') \in E \wedge
\rho(s)\models g\}\enspace$.

\[
\irule
{
\vflow\com
v\not \in \Dom{h}
}
{
\mtrans{\alpha,\rho(0)}{\rho, \theta, \omega}{\alpha[v],\rho(t)}
}\label{rule:time:automaton:target-h-nd}
\]

In \oCIF, if an automaton performs an environment transition then an
unique active location is chosen. When hierarchy is incorporated, the
substructure is initialized as well. This is expressed by
Rule~\ref{rule:environment:automaton}. The initialized composition $p$
becomes the active substructure of $\alpha[v]$, and the automaton is
terminating if the location and the active substructure are.
Rule~\ref{rule:environment:automaton-target-nd} deals with the case
where there is no substructure.

\begin{mathpar}
\irule
{
\sigma\models\initf(v),
\ctrans{h(v),\sigma}{b}{p,\sigma'}
}
{
\ctrans{\alpha,\sigma}{\sigma\models\termf(v)\wedge b}{p:\alpha[v],\sigma'}
}\label{rule:environment:automaton} \quad

\irule
{
\sigma\models\initf(v),
v\not\in\Dom{h}
}
{
\ctrans{\alpha,\sigma}{\sigma\models\termf(v)}{ \alpha[v],\sigma'}
}\label{rule:environment:automaton-target-nd}
\end{mathpar}

\subsection{Automaton postfix operator}\label{subsec:aut-post-op}

We now define the SOS rules for the automaton postfix operator, which
helps in defining the overall behavior of a hierarchical automaton.

Intuitively, the composition $p:\alpha$ means that composition $p$ is
the active substructure of some initial location $v\in V$ in the
automaton $\alpha$. Note, that whenever the automaton postfix is only used 
as an auxiliary operator, this initial location will always be uniquely 
specified. 

If we now look at the structure of a state $(p,\sigma)$ in the hybrid transition system,
it becomes clear how the postfix operator helps us to mimick the state-tree structures used
in the semantics of statecharts \cite{Harel87-visformalism}. Figure \ref{fig:postfix-operator-stss} shows that a 
composition $p$ in essence \emph{is} a tree, where the postfix operator represents the
edges of the tree and the parallel compositions represents the branching. The root of
this tree is the active location of the automaton we described, while the leaves are
the active substructures where the control over actions currently lies. Indeed, an informal
comparison of our semantics to that of statecharts suggest that that the AND-superstates of statecharts
are represented as (asynchronous) parallel compositions $\parallel_{\emptyset}$, 
while OR-superstates are represented by having multiple locations for which the initialization predicate holds.

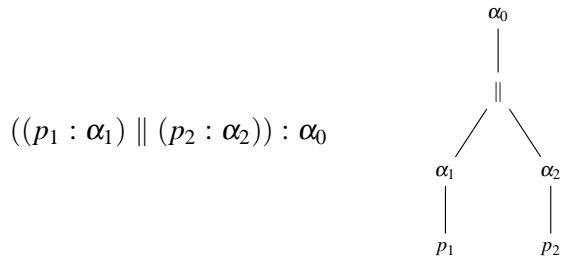
\begin{figure}[htb]
  \centering
  \begin{tabular}{c c}
    \begin{minipage}{0.3\linewidth}
      $((p_1:\alpha_1) \parallel (p_2:\alpha_2)):\alpha_0$
    \end{minipage}
    &
    \begin{minipage}{0.3\linewidth}
      \scalebox{0.7}{
        \tikzset{
  super/.style={
    rectangle,
    rounded corners,
    draw=black,
  }
}

\tikzset{
  active/.style={
    draw=red,
  }
}

\tikzset{
  state/.style={
    circle,
    draw=black,
  }
}

\newcommand{\spce}{1.2}

\begin{tikzpicture}[-,>=stealth',shorten >=1pt,auto, semithick]

\node (atop) {$\alpha_0$};
\node (pcomp) at ($(atop.south)+(0,-\spce)$) {$\parallel$};
\node (aright) at ($(pcomp.south)+(1,-\spce)$) {$\alpha_2$};
\node (aleft) at ($(pcomp.south)+(-1,-\spce)$){$\alpha_1$};
\node (pright) at ($(aright.south)+(0,-\spce)$) {$p_2$};
\node (pleft) at ($(aleft.south)+(0,-\spce)$) {$p_1$};

\path (atop) edge (pcomp);
\path (pcomp) edge (aleft);
\path (pcomp) edge (aright);
\path (aleft) edge (pleft);
\path (aright) edge (pright);

\end{tikzpicture}
%%% Local Variables:
%%% mode: latex
%%% TeX-master: "../scif_ppt"
%%% End:

%%% Local Variables: 
%%% mode: latex
%%% TeX-master: "../text"
%%% End: 
}
    \end{minipage}
  \end{tabular}
  \caption{Relation between automaton postfix operator and state tree structures}
  \label{fig:postfix-operator-stss}
\end{figure}

The semantics of $p:\alpha$ is reminiscent of the sequential composition of untimed process algebra,
i.e. the composite automaton $p$ in $p:\alpha$ will perform action transitions until it is terminating, 
after that the automaton $\alpha$ can perform its action transitions. The difference between the
sequential composition operator and the automaton postfix operator is due to the difference in the 
passage of time caused by these operators. In the automaton postfix operator, the passage of time is
synchronized between the first and second component, whereas in sequential composition the passage of 
time is not synchronized. There the second component waits for termination of the first, in a similar
way as with action transitions.

All the rules presented here are similar to the those presented in the
previous section. The difference lies in that in the rules of this
section function $h$ is not considered, since there is an active
substructure $p$ in the target state of every transitions appearing
in the conclusions. Rule~\ref{rule:action:postfix-operator:sa-term}
models the action transition when the substructure is terminating.
\[
\irule
  {
    \eedge,
    \ctrans{p,\sigma}{\true}{p',\sigma},
    \ctrans{h(v'),\sigma'}{b}{q,\sigma'}
  }
  {
    \trans{p:\alpha,\sigma}{a}{q:\alpha[v'],\sigma'}
  }\label{rule:action:postfix-operator:sa-term}
\]

Rule~\ref{rule:action:po:target-h-nd} also models the action
transition generated from a postfix operator when the substructure is
terminating and the target location $v'$ does not contain a
substructure. Rule~\ref{rule:action:postfix-operator:sa-nterm} models
the action transition which is a result of the execution of the
substructure.

\begin{mathpar}
\irule
  {
    \eedge,
    \ctrans{p,\sigma}{\true}{p',\sigma},
    v'\not\in\Dom{h}
  }
  {
    \trans{p:\alpha,\sigma}{a}{\alpha[v'],\sigma'}
  }\label{rule:action:po:target-h-nd}
\quad

\irule
  {
    \trans{p,\sigma}{a}{q,\sigma'}
  }
  {
    \trans{p:\alpha,\sigma}{a}{q:\alpha,\sigma'}
  }\label{rule:action:postfix-operator:sa-nterm}
\end{mathpar}

Finally, Rule~\ref{rule:time:postfix-operator} models the passage of
time in an automaton postfix such that the timed transitions are
synchronized in every level of hierarchy $p:\alpha$, where,
$\Dom{\omega}=\Dom{\rho}\com \Dom{\theta} = \Dom{\rho}$ and
$\forall_{s \in [0, t]}. \omega(s) \triangleq \omega_0(s) \wedge
\rho(s)\models\termf(v)$, and $\forall_{ s\in[0,t]}.\theta(s)\triangleq
\theta_0(s) \cup \{a \mid (v,g,a,r,v) \in E \wedge \rho(s)\models g
\wedge \omega_0(s)\}$.

\[
\irule
  {
    \vflow,
    \mtrans{p,\rho(0)}{\rho, \theta_0,\omega_0}{p',\rho(t)}
  }
  {
    \mtrans{p:\alpha,\rho(0)}{\rho, \theta, \omega}{p':\alpha[v],\rho(t)}
  }\label{rule:time:postfix-operator}
\]

Rule~\ref{rule:environment:postfix-operator} models the execution of environment transition in an automaton postfix.
\[
\irule
  {
    \ctrans{p,\sigma}{b}{p',\sigma'}
  }
  {
    \ctrans{p:\alpha,\sigma}{\sigma\models\termf(v)\wedge b}{p':\alpha[v],\sigma'}
  }\label{rule:environment:postfix-operator}
\]

\subsection{Parallel composition operator}\label{sem-paras}

The parallel composition operator allows synchronisation of equally labeled action transitions between any two components that are specified by the synchronisation set $S\subseteq \actions$. Rule~\ref{rule:action:sync-parallel} models this fact.

\[
\irule
{
\trans{p,\sigma}{a}{p',\sigma'},
\trans{q,\sigma}{a}{q',\sigma'},
a\in S
}
{
\trans{p \paras q,\sigma}{a}{p' \paras q',\sigma'}\\\\
\trans{q \paras p,\sigma}{a}{q' \paras p',\sigma'}
}\label{rule:action:sync-parallel}
\]

Rule~\ref{rule:action:interleaving-parallel-p} models the interleaving of action transitions that do not belong to the synchronisation set $S$. Note the presence of environment transition in the premise of the following rule. This allows the other component (which does not perform an action transition, in the following case $q$) to get initialised.

\[
\irule
{
\trans{p,\sigma}{a}{p',\sigma'},
\ctrans{q,\sigma}{b}{q',\sigma'},
a\not\in S
}
{
\trans{p \paras q,\sigma}{a}{p' \paras q',\sigma'}\\\\
\trans{q \paras p,\sigma}{a}{q' \paras p',\sigma'}
}\label{rule:action:interleaving-parallel-p}
\]

In a parallel composition, time can pass it is can pass in each
component individually, as it can be seen in
Rule~\ref{rule:time:parallel}, where $\forall s\in[0,t].
\big[\theta_{01}(s)=\left(\theta_0(s)\cap\theta_1(s)\right) \cup
\left(\theta_0(s)\setminus S\right) \cup \left(\theta_1(s)\setminus
  S\right)\big]$, and $\forall
s\in[0,t].[\omega_{01}(s)=\omega_0(s)\wedge \omega_1(s)]$.
The guard trajectory is constructed (same as in \oCIF) from the guard
trajectories of the composite automata interleaving in the parallel
composition: at a given time point $s$, an action is enabled in the
parallel composition $p \paras q$ if it is enabled in $p$ and $q$
(regardless of whether the action is in $S$), or if it is enabled in
$p$ or $q$ and is not synchronizing in the other component. The
termination trajectory in the parallel composition at a given time
point $s$ is the conjunction of the termination trajectories of the
respective components at the same time instant $s$.

\[
\irule
{
\mtrans{p,\rho(0)}{\rho,\theta_0,\omega_0}{p',\rho(t)},
\mtrans{q,\rho(0)}{\rho,\theta_1,\omega_1}{q',\rho(t)}
}
{
\mtrans{p\paras q,\rho(0)}{\rho,\theta_{01},\omega_{01} }{p'\paras q',\rho(t)}\\\\
\mtrans{q\paras p,\rho(0)}{\rho,\theta_{01},\omega_{01} }{q'\paras p',\rho(t)}
}\label{rule:time:parallel}
\]

The initialization of a parallel composition
(Rule~\ref{rule:environment:parallel}) is the initialization of its
components. The termination predicate in the parallel composition is
the conjunction of the termination predicates of the respective
components.

\[
\irule
{
\ctrans{p,\sigma}{b_0}{p',\sigma'},
\ctrans{q,\sigma}{b_1}{q',\sigma'}
}
{
\ctrans{p\paras q,\sigma}{b_0\wedge b_1}{p' \paras q',\sigma'}\\\\
\ctrans{q\paras p,\sigma}{b_0\wedge b_1}{q' \paras p',\sigma'}
}\label{rule:environment:parallel}
\]

\subsection{Stateless bisimulation}\label{subsec:stateless-bisim}

It is clear from the definition of a HTS (Definition~\ref{def:hts}) that a state in a transition system consists of a process part (a behavioural entity) and a data part (valuation). Furthermore, we know that the stateless bisimulation is the most robust equivalence for the transition systems whose states contains data \cite{MousaviRenGro}. This subsection shows that the semantics of \mCIF is compositional with respect to stateless bisimulation \cite{MousaviRenGro}, i.e. stateless bisimulation is a congruence for all operators of $\mCIF$.
%\begin{proposition}
%$(\ASpec\times\Sigma,\actions,\xrightarrow{},\longmapsto, \dashrightarrow)$ is the hybrid transition system induced by the constructs of \mCIF.
%\end{proposition}

\begin{definition} \label{def:stateless-bisim}
A symmetric relation $R\subseteq\ASpec\times\ASpec$ is called a
stateless bisimulation \cite{MousaviRenGro} relation iff the following transfer conditions hold.
\begin{itemize}
\item $\forall p,p',\sigma,\sigma',a,q.
    \Big[\trans{p,\sigma}{a}{p',\sigma'} \wedge (p,q)\in R \Rightarrow \exists q'.\big[\trans{q,\sigma}{a}{q',\sigma'}\wedge (p',q')\in R\big]\Big]\enspace.$
\item $\forall p,p',\sigma,\sigma',\rho,\theta,\omega,q.
    \Big[\mtrans{p,\sigma}{\rho,\theta,\omega}{p',\sigma'} \wedge (p,q)\in R \Rightarrow \exists q'.\big[\mtrans{q,\sigma}{\rho,\theta,\omega}{q',\sigma'}\wedge (p',q')\in R\big]\Big]\enspace.$
\item $\forall p,p',\sigma,\sigma',b,q.
    \Big[\ctrans{p,\sigma}{b}{p',\sigma'} \wedge (p,q)\in R \Rightarrow \exists q'.\big[\ctrans{q,\sigma}{b}{q',\sigma'}\wedge (p',q')\in R\big]\Big]\enspace.$
\end{itemize}
Two composite automata $p,q$ are said to be stateless bisimilar
(denoted $p\sbisim q$) iff there exists a stateless bisimulation relation $R$ such that $(p,q)\in R$.
\end{definition}

\begin{theorem}\label{thm:congruence}
Stateless bisimulation is a congruence for all the constructs of \mCIF.
\begin{proof}
The SOS rules of \mCIF are in the \textit{process-tyft} format, which guarantees the congruence for stateless bisimilarity \cite{MousaviRenGro}.
\end{proof}
\end{theorem}

\section{Elimination of hierarchy}\label{sec:elim-hierarchy}

In this section, we present a technique that converts a hierarchical
automaton into an automaton in which the hierarchical function $h$ is
empty, such that they are stateless bisimilar. Such techniques in general, are known as \emph{linearization} or \emph{elimination} of operators\cite{Usenko}. The advantage of \emph{flattening} a hierarchical automaton is that it allows the reuse of existing tools like simulators of the \oCIF language, which are only developed for flat automata. % A disadvantage of flattening, is that (so far) we only know how to do it for hierarchical automata that are, in a sense, well founded.

\begin{definition}
The \emph{depth} $D(p)$ of an composite automaton $p \in \ASpec$ is recursively defined to be the 1 + $\max_{v \in \dom(h)} \ D(h(v))$ when $p$ is a hierarchical automaton of the form $(V,\initf,\tcpf,E,\termf,h)$, and is defined as $\max (D(q),D(r))$ whenever $p = q \paras r$. An automaton is called \textit{well founded} whenever its depth is defined. An automaton of depth $1$ (i.e. with $\dom(h) = \emptyset$) is also called a \textit{flat} automaton.
\end{definition}

Suppose that we have a procedure $\enc$ that turns any composition of flat automata into a stateless bisimilar flat automaton. In particular, suppose that $\enc(\alpha \paras \alpha') \sbisim \alpha \paras \alpha'$ whenever $\alpha$ and $\alpha'$ are flat automata. Then, we can lift this procedure to any well-founded composite automaton $p \in\ASpec$, by first applying it to all components of $p$ before applying it to $p$ itself. We define $\enc(p \paras q) = \enc(\enc(p) \paras \enc(q))$ and $\enc((V,\initf,\tcpf,E,\termf,h)) = \enc((V,\initf,\tcpf,E,\termf,\tilde{h}))$ with $\tilde{h}(v) = \enc(h(v))$, for any $p$, $q$ and $(V,\initf,\tcpf,E,\termf,h)$ of depth greater than $2$. Structural induction on the depth of the composite automaton, combined with the congruence obtained in theorem \ref{thm:congruence}, then gives us $\enc(p) \sbisim p$ for all well-founded composite automata $p$.

Such a procedure is already known for all the usual operations of the CIF\cite{tim}, and next, we will give it for hierarchical automata.

\begin{definition}\label{def:flattening}
Let $\alpha\in\ASpec$ be an automaton of depth $2$ of the form $(V,\initf,\tcpf,E,\termf,h)$, such that $h(v)$ is a flat automaton for all $v \in \dom(h)$. We define $\enc(\alpha)=(\hat{V},\finitf,\ftcpf,\hat{E},\ftermf,\emptyset)$ where,
\begin{itemize}
\item The set $\hat{V}$ of locations of the flat automaton $\enc(\alpha)$ is defined by:
\[\hat{V}\triangleq
\bigcup_{v\in V} \left\{(v,w)\Big| \begin{array}{l}
\big(v\not\in\Dom{h}\wedge w=\bot\big) \quad\vee\\
\big(h(v)=\ptuple{}{\emptyset} \wedge w\in V'\big)
\end{array} \right\}
\]

\item The predicate-functions $\hat{\Box}$, with $\Box\in\{\initf,\tcpf,\termf\}$, are defined for each $\hat{v}\in\hat{V}$ by:
\[\hat{\Box}(\hat{v})\triangleq\left\{
                                 \begin{array}{ll}
                                   \Box(v), & \hbox{if $\hat{v}=(v,\bot)$} \\
                                   \Box(v)\wedge \Box'(w), & \hbox{if $\left(
\hat{v}=(v,w)\wedge h(v)=\ptuple{}{\emptyset}\wedge
w\in V'\right)$}
                                 \end{array}
                               \right.\]

\item The edges $\left(\hat{v}_0,g',a,r',\hat{v}_1\right)$ of the flat automaton $\enc(\alpha)$ are present (i.e. $\left(\hat{v}_0,g',a,r',\hat{v}_1\right)\in \hat{E}$) iff one of the following conditions hold:
\begin{enumerate}
%\item $\hat{v}_0=(v_0,\bot)\wedge \hat{v}_1=(v_1,\bot)$ for some $v_0,v_1\in V$ such that, \[(v_0,g',a,r',v_1)\in E\enspace.\]
%\item $\hat{v}_0=(v_0,w_0) \wedge \hat{v}_1=(v_1,\bot)$ for some $v_0,v_1\in V$ such that,
%    \begin{eqnarray*}
%    && \enc(h(v_0))=\ptuple{}{\emptyset} \wedge (v_0,g,a,r',v_1)\in E \wedge \\
%    && w_0\in V'\wedge g'=(\termf'(w_0)\wedge g)\enspace.
%    \end{eqnarray*}
%\item $\hat{v}_0=(v_0,\bot) \wedge \hat{v}_1=(v_1,w_1)$ for some $v_0,v_1\in V$ such that,
%    \begin{eqnarray*}
%    && \enc(h(v_1))=\ptuple{}{\emptyset} \wedge (v_0,g',a,r,v_1)\in E \wedge\\
%    && w_1\in V'\wedge r'=(r\wedge\initf'(w_1)^{+})\enspace.
%    \end{eqnarray*}
\item $\hat{v}_0=(v_0,w_0) \wedge \hat{v}_1=(v_1,w_1)$ for some $v_0,v_1\in V$ such that $h(v_i)=\ptuple{i}{\emptyset}$, $w_i\in V'_i$, for $i\in\{0,1\}$ and
    \begin{eqnarray*}
    && (v_0,g,a,r,v_1)\in E \wedge g'=\left(\termf'_0(w_0)\wedge g\right) \wedge r'=\left(r\wedge \initf'_1(w_1)^{+}\right)\enspace.
    \end{eqnarray*}
    Note that if $w_0=\bot$ ($w_1=\bot$) then by defining $\termf'_0(w_0)=\true$ ($\initf'_1(w_1)=\true$) we get definitions for the simpler cases derived from the above one.
\item $\hat{v}_0=(v,w_0)\wedge \hat{v}_1=(v,w_1)$ for some $v_0\in V$ such that
    \begin{eqnarray*}
    && h(v)=\ptuple{}{\emptyset} \wedge w_0,w_1\in V'\wedge (w_0,g',a,r',w_1)\in E' \enspace.
    \end{eqnarray*}
\end{enumerate}
\end{itemize}
\end{definition}

% Example
Figure~\ref{fig:flat-thermostat-clock} shows the resulting automaton
after applying the linearization procedure to the thermostat model
extended with a clock
(Figure~\ref{fig:thermostat-clock}).
\vspace{-1ex}

\begin{figure}[htb]
  \centering
  \scalebox{0.7}{
    \tikzset{
  super/.style={
    rectangle,
    rounded corners,
    draw=black,
  }
}

\tikzset{
  active/.style={
    draw=red,
  }
}

\tikzset{
  state/.style={
    circle,
    draw=black,
  }
}
\begin{tikzpicture}[->,>=stealth',shorten >=1pt,auto, semithick]

\node (i) {};
\node (off) [state] at ($(i.east)+(2.7,0)$) {$\begin{array}{c}
\mathit{Off} \\ \dot{T}=-T+15
\end{array}$};
\node (oncold) [state] at ($(off.center) + (10,2)$) {$\begin{array}{c}
\mathit{On},\mathit{Cold} \\ \dot{T}=-T+25 \\ \wedge \dot{c} = 1
\end{array}$};
\node (onhot) [state] at ($(off.center) + (10,-2)$) {$\begin{array}{c}
\mathit{On},\mathit{Hot} \\ \dot{T}=-T+25
\end{array}$};

\path (i) edge node {$\begin{array}{c}
T=25\, \wedge\\ n=0
\end{array}$} (off);
\path (off) edge  [bend left=12] node [above=20] {$T < 20 :\textit{switch-on}:
  n^{+}=n+1 \wedge c^+=0$} (oncold);
\path (onhot) edge [bend left=12] node [below=20] {$ n \leq 1000 :\textit{switch-off}:
  \True$} (off);
\path (oncold) edge [bend left=12] node {$ \Delta \leq t :\textit{done}:
  \True$} (onhot);

\end{tikzpicture}
%%% Local Variables:
%%% mode: latex
%%% TeX-master: "../m-HCIF"
%%% End:
}
  \caption{A flat model of thermostat with
    clock.}\label{fig:flat-thermostat-clock}
\end{figure}
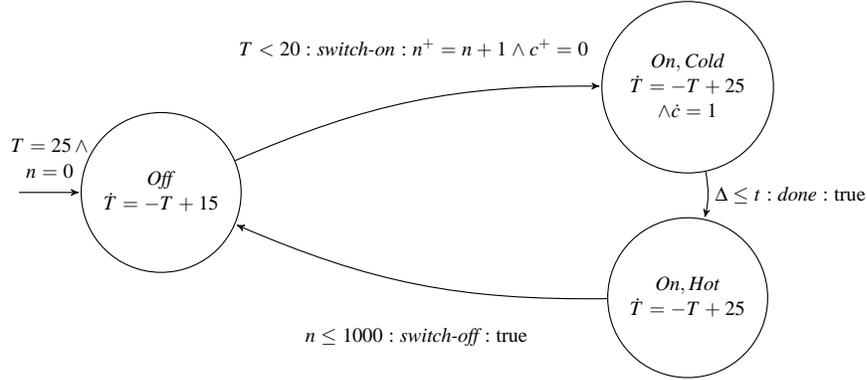

Next we prove the correctness of our linearization procedure as defined in Definition~\ref{def:flattening}.

\begin{theorem}\label{thm:elim}
Let $\alpha\in\ASpec$ be an automaton of depth $2$ of the form $(V,\initf,\tcpf,E,\termf,h)$, such that $h(v)$ is a flat automaton for all $v \in \dom(h)$, then $\alpha \sbisim \enc(\alpha)$.
\end{theorem}
\begin{proof}
Fix $\hat{\alpha} = \enc(\alpha)=(\hat{V},\finitf,\ftcpf,\hat{E},\ftermf,\emptyset)$. It is rather straightforward but tedious to verify that the relation $R = \{ \big(\alpha,\hat{\alpha}\big), \big(\alpha[v],\hat{\alpha}[(v,\bot)]\big), \big(h(v)[w]:\alpha[v],\hat{\alpha}[(v,w)]\big)\mid v \in V \ \wedge (v,w) \in \hat{V} \}$, is a witnessing stateless bisimulation.
\end{proof}

\section{Conclusions}\label{sec:conclusion}
In this paper we illustrated how to add hierarchy to a subset of the
CIF (called \mCIF) in a compositional way, and we showed that the SOS
rules of atomic entities can be modified without altering the rules of
the CIF operators.
Moreover, the rules are formulated in such a way that the addition of
concepts such as urgency and invariants can be incorporated easily
without altering the rules presented here.
However, the usability of \mCIF is not yet investigated and the plan is to evaluate it by performing industrial case-studies within the context of the MULTIFORM project \cite{multiform} after extending it with the remaining operators of the CIF. A procedure to eliminate hierarchy was given in order to be able to use the existing tools associated with CIF. Note that Definition~\ref{def:flattening} presented here is a relatively inefficient one to implement. It can be further optimized, for example, by disallowing the edges in the set $\hat{E}$ that are never executed for any valuation.

As ongoing work, we are researching a branching version of stateless
bisimulation for hybrid transition systems in order to handle $\tau$ action as `invisible'. Using such a notion, it becomes possible to formalize when a stepwise refinement is a correct refinement of an abstract model. Thus we can attack the second aspect of stepwise refinement mentioned in the introduction.

\textbf{Acknowledgements}: The authors would like to thank Jos Baeten,
Koos Rooda and Mohammad Mousavi,  and the reviewers of SOS, for their
constructive comments that very much helped to improve this paper.

This work has been performed as part of the "Integrated Multi-formalism Tool Support for the Design of Networked Embedded Control Systems (MULTIFORM) project, supported by the Seventh Research Framework Programme of the European Commission.
Grant agreement number: INFSO-ICT-224249.

\bibliographystyle{plain}
\bibliography{hybchi}  % sigproc.bib is the name of the Bibliography in this case

\end{document}